\documentclass[runningheads]{llncs}
\usepackage[T1]{fontenc}
\usepackage{tabularx}
\usepackage{multirow}
\usepackage{bm}
\usepackage[hidelinks]{hyperref}
\usepackage{subfigure}
\usepackage{cite}
\usepackage{color}
\usepackage{svg}
\usepackage{amsmath}
\usepackage{graphicx}
\usepackage[linesnumbered,ruled,vlined]{algorithm2e}
\usepackage{amssymb}
\usepackage{adjustbox}
\usepackage{siunitx}
\usepackage{enumerate}
\newcommand{\stitle}[1]{\vspace*{0.4em}\noindent{\bf #1.\/}}

\definecolor{cosmiclatte}{rgb}{1.0, 0.97, 0.91}

\begin{document}
%
%\title{Border Labeling for Distance Query}
\title{Exploring Distance Query Processing in Edge Computing Environments}
%
%\titlerunning{Abbreviated paper title}
% If the paper title is too long for the running head, you can set
% an abbreviated paper title here
%
\author{Xiubo Zhang\inst{1}\and
Yujie He\inst{3}
\and Ye Li\inst{2}
\and Yan Li\inst{4}
\and  Zijie Zhou\inst{1}\and Dongyao Wei\inst{1} \and Ryan\inst{1}}
\authorrunning{F. Author et al.}
% First names are abbreviated in the running head.
% If there are more than two authors, 'et al.' is used.
%
\institute{University of Macau, Macau    \and Shenzhen Institute of Beidou Applied Technology \and
Macau University of Science and Technology \and 
Shenzhen Polytechnic University\\
} 

\maketitle              % 
\begin{abstract}
In the context of changing travel behaviors and the expanding user base of Geographic Information System (GIS) services, conventional centralized architectures responsible for handling shortest distance queries are facing increasing challenges, such as  heightened load pressure and longer response times. To mitigate these concerns, this study is the first to develop an edge computing framework specially tailored for processing distance queries. In conjunction with this innovative system, we have developed a straightforward, yet effective, labeling technique termed Border Labeling. Furthermore, we have devised and implemented a range of query strategies intended to capitalize on the capabilities of the edge computing infrastructure. Our experiments demonstrate that our solution surpasses other methods in terms of both indexing time and query speed across various road network datasets.  The empirical evidence from our experiments supports the claim that our edge computing architecture significantly reduces the latency encountered by end-users, thus markedly decreasing their waiting times.
\end{abstract}
%
%
%\\

\section{Introduction}\label{sec:introduction}

Distance query plays an important role in Geographical Information Systems (GIS) and is widely utilized in spatial analysis. It involves finding the shortest distance on a road network, which is a common task in our daily life. According to reports from the Chinese government\footnote{\url{https://www.gov.cn/xinwen/2022-08/29/content_5707349.htm}} and the China Information Association\cite{Information}, China's enterprise location service platforms, such as Baidu Maps, Amap, Tencent Location, and Huawei Maps, process 130 billion location service requests per day. Also, Didi, a leading ride-hailing company, has 493 million users as of 2022\cite{website:didi}. Service providers encounter the challenge of efficiently handling a large volume of queries while considering real-time road network updates. Our work focuses on alleviating the burden of query requests and prioritizing swift responses to address these concerns.

Previous studies on distance query can be divided into two main categories: online search \cite{DIJKSTRA1959,A*} and bidirectional search \cite{bi-direct,bi-di2,DBLP:conf/alenex/bi-di3}. Both categories rely on centralized services that run on a single server, which makes them inefficient for processing a large number of queries in a short period of time, as shown in Fig.~\ref{fig:query}. Another challenge is to handle queries that incorporate the dynamic road network information. Current studies \cite{UE, DBLP:conf/icde/ZhangLHMCZ21, DBLP:conf/sigmod/WeiWL20,dynamic1} consider the dynamic changes in road networks and point out the challenges of index updating. If the index is not updated in time, users may have to use outdated road network information or wait longer for the query results, which may affect their user experience or travel decisions.

To supplement existing centralized methods for processing distance queries, we proposed an edge computing framework that is specifically designed to process distance queries. Notably, the framework is accompanied by an effective labeling technique called Border Labeling. Additionally, we point out that a variety of query strategies that take advantage of the capabilities of the edge computing infrastructure can be applied.

Our contributions are described as follows: (1) We present a novel system that operates within an edge computing environment to address the task of answering distance queries on dynamic graphs; (2) We create a simple yet powerful labeling technique called {\em Border Labeling}; (3) We propose a local bound for local distance queries, which reduces the overall response time while still secures the shortest path distance.

In Section~\ref{sec:hl}, we present the prerequisites of our algorithm: the hub labeling algorithm based on hub pushing, a pruning scheme for it, as well as the concepts of borders and districts. In Section \ref{sec:method:BL}, we formally introduce the process of our border labeling algorithm and prove its correctness. In Section~\ref{sec:enivornments}, we explain how we adapt border labeling for use in an edge computing environment. Our algorithm's performance is evaluated in terms of construction and query processing in Section~\ref{sec:exp}, and we also assess our edge computing framework's efficacy in scenarios characterized by high-frequency and voluminous updates of road network information. Finally, we conclude our findings in Section~\ref{conclusion}.

\section{Hub Labeling}\label{sec:pre}
\label{sec:hl}

A hub label, represented as $L(v)$, is a set containing some pairs of the form ${\langle}u,d_G(u,v){\rangle}$, where $u$ and $v$ are vertices, and $d_G(u,v)$ denotes the shortest distance from vertex $u$ to vertex $v$ in the graph $G$. The complete set of labels is denoted as $L = \bigcup_v L(v)$, which encompasses the union of all hub labels for every vertex in the graph. The primary goal of constructing hub labels is to ensure the 2-hop coverage property (for query processing) while minimizing the size of the label set $L$.

\subsection{Constructing Hub Labels by Hub Pushing}

\begin{figure}[!htb]
    \centering
    \subfigure[Pushing hub $v_0$ to all vertices. The distance from $v_0$ to each vertex is kept in the label set of each vertex.]{\includegraphics[width=0.48\columnwidth]{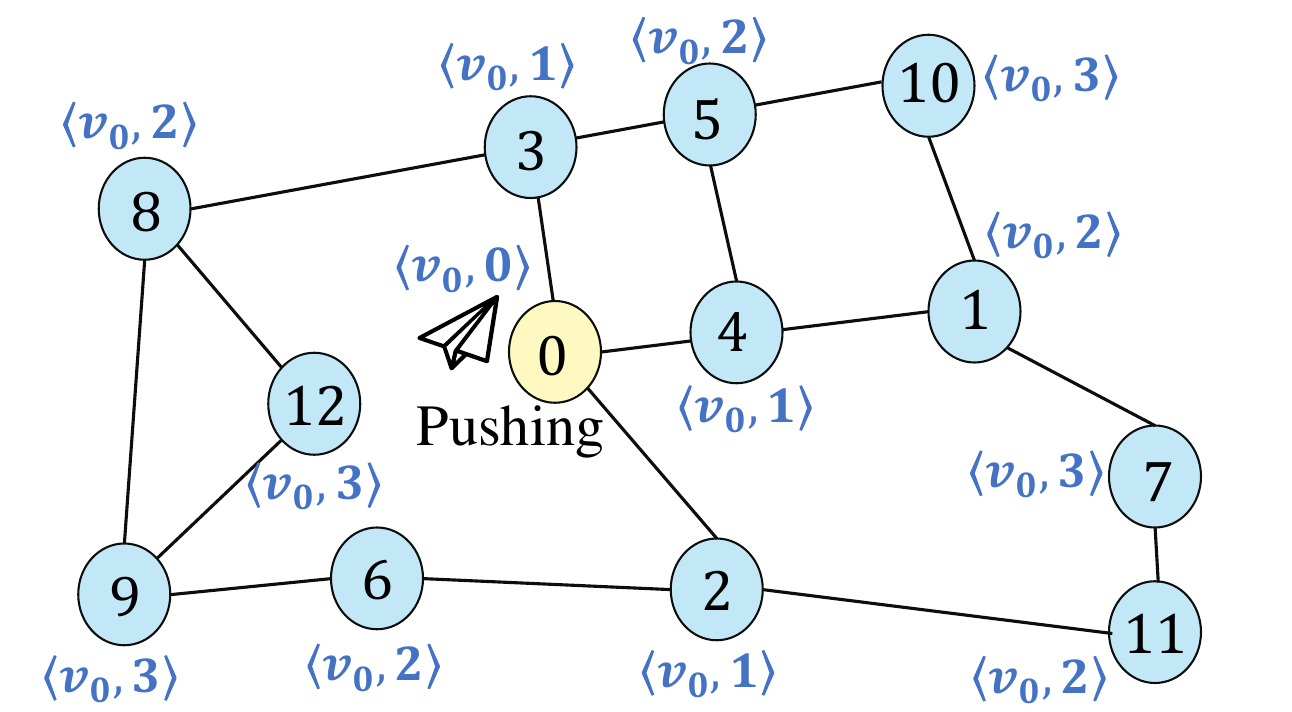}} \hspace{2mm}
    \subfigure[Pushing hub $v_1$. We pruned $v_2$ and $v_3$ and stopped traversing from them. We do not push $v_1$ to $v_0$ for its priority.]
    {\includegraphics[width=0.48\columnwidth]{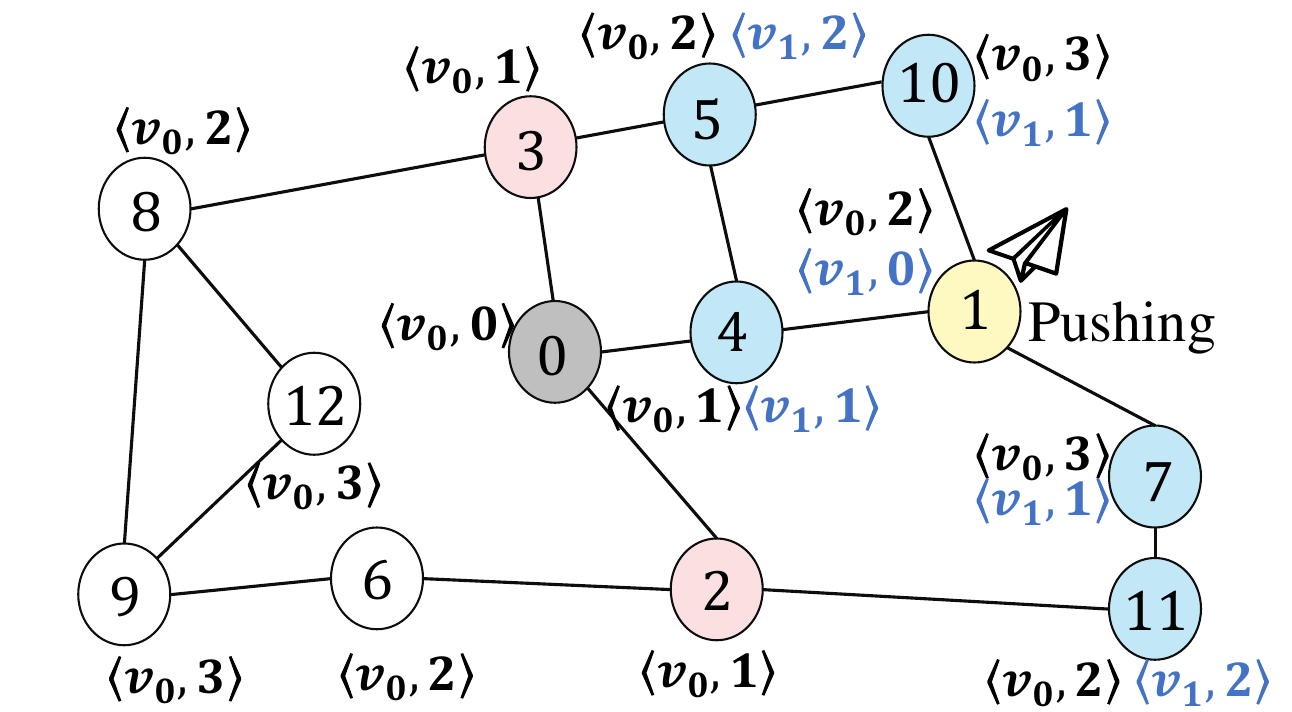}}
    
\caption{Example hub pushing situations. In each case, the yellow vertex denotes the root of one pushing operation, the blue vertices denote those which were visited and labeled in the pushing, the pink vertices denote those which were visited but pruned, and the gray vertices denote roots of previous pushing operations.}
\label{fig::1b}
\label{fig:result_of_DCA}
\end{figure}

A commonly used technique for constructing the index $L$ is an iterative method known as {\em hub pushing}~\cite{PLL13}. This approach involves pushing a vertex label, referred to as a hub, to all vertices that it can reach with a higher hierarchy in the ordering. 
The concept of {\em order} $\mathcal{O}$ in this context signifies the precedence in which vertices are selected for label pushing. Specifically, vertices assigned with lower order values are given precedence and their labels are pushed first. 
Once all vertices have been processed in this priority-driven manner, the label structure is considered complete.

\begin{example}
To provide an illustration, let us consider the vertex with the highest priority, referred to as $v_1$, within the graph shown in Figure~\ref{fig::1b}(a). During the hub pushing phase, labels in the form \(\langle v_1,d_G(v_1,u)\rangle\) are pushed to all reachable vertices, thereby incorporating $v_1$ into their respective labels as a hub.
\end{example}

\stitle{Query Processing} Given a label set $L$, we can find the distance between two vertices $s$ and $t$ by a linear join process as follows.

\begin{definition}[Hub Label Distance Query $\lambda(s,t,L)$]\label{def:DQ}
The distance between vertices $s$ and $t$ can be obtained by applying a linear join of labels from the label set $L$.
\begin{equation}
   \lambda(s,t,L) = \min_{h \in L(s) \cap L(t)}(d_G(s,h) + d_G(h,t))
\end{equation}

\end{definition}

We define $\lambda(s,t,L) = \infty$ if $L(s)$ and $L(t)$ do not have any common vertex. It is important to note that the correctness of the distance calculation relies on the label set $L$ satisfying the 2-hop cover property (Definition~\ref{def:2hopcover}). 

\begin{definition} [2-Hop Cover] \label{def:2hopcover}
We call a set of labels $L$ a \textit{2-hop cover} of $G$ if $d_G(s,t) = \lambda(s,t,L)$ for any pair of vertices $s$ and $t$.
\end{definition}

\stitle{Pruned Landmark Labeling} Akiba et al. \cite{PLL13} proposed  a pruning method of the naive hub labeling. When performing a Dijkstra starting from vertex $v$ and visiting vertex $u$, with a partially built label set $L$, if $\lambda(v, u, L) = d_G(v, u)$, label ${\langle}v,d_G(u,v){\rangle}$ will not be added to the modified label set $L'(u)$ and the algorithm stops traversing any edge from vertex $u$. This is because distance $d_G(v,u)$ in graph $G$ can already be determined by combining the stored pairs in labels $\langle w,d_G(v,w)\rangle\in L(v)$ and $\langle w,d_G(w,u)\rangle\in L(u)$, which indicates that labels in vertex $w$ has already provided sufficient information for computing the shortest distance. In other words, any further traversal of an edge from vertex $u$ would only lead to unnecessary redundancy in storage. They named this technique \textit{Pruned Landmark Labeling (PLL)}. For example, in Figure~\ref{fig::1b}(b), pruning happened at $v_2$ and $v_3$ in the pushing from $v_1$.

\subsection{Decomposition-based Hub Labeling}

By decomposing a graph into multiple smaller subgraphs, it becomes feasible to convert a distance query into a series of sub-distance queries. This decomposition can significantly reduce the search space and memory usage required for evaluating the subproblems. Various methods \cite{par16,par2} have showcased the advantages of such techniques in terms of scalability, particularly in reducing search space and memory consumption for handling large datasets. The decomposition process creates two new concepts, {\em district}  and {\em border}.

\begin{figure}[ht]
  \centering

  \subfigure{\adjustbox{valign=c}{\includegraphics[width=0.45\textwidth]{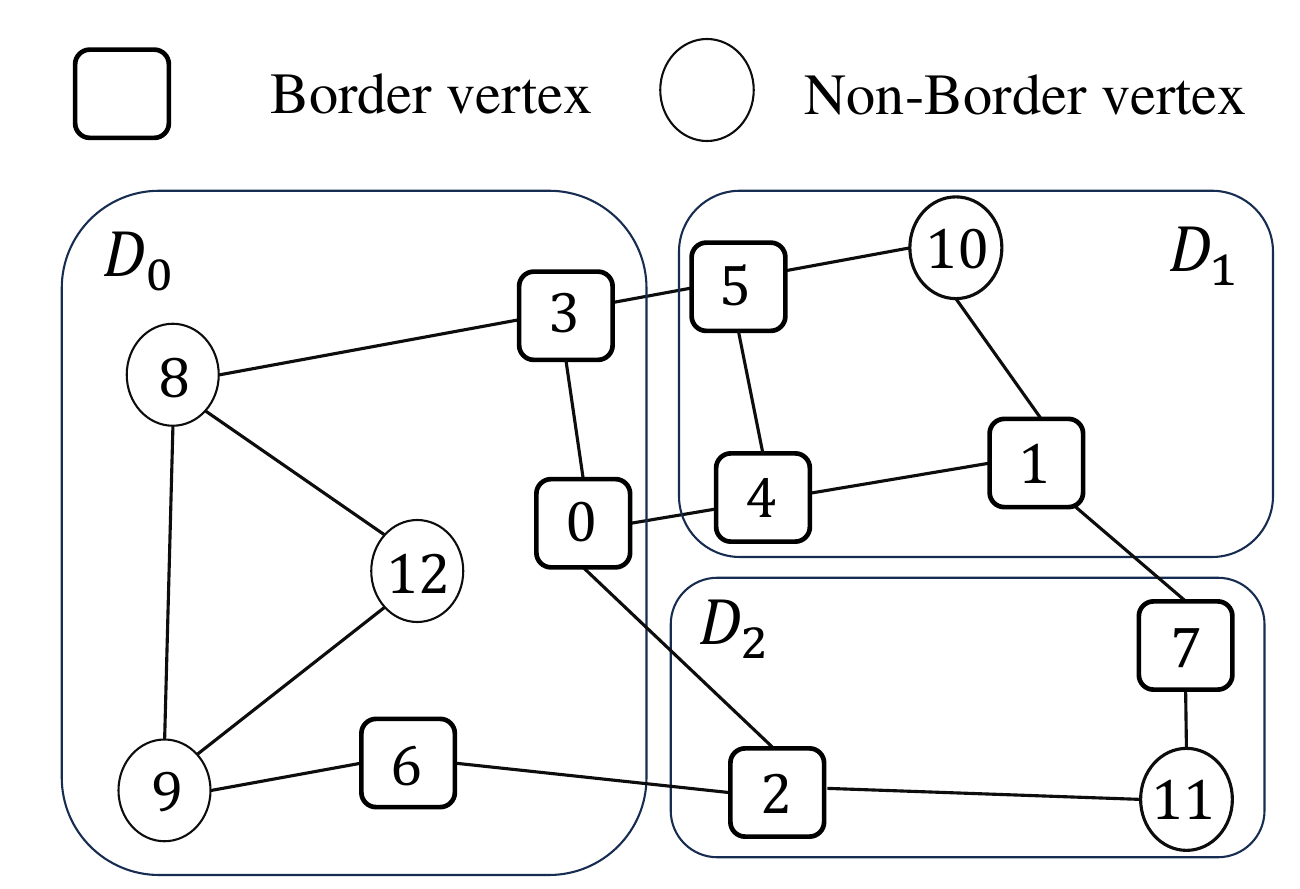}} 
  \label{fig:example1}
  }
  \subfigure{
  \scriptsize
  \begin{tabular}{|c|c|l|}
      \hline $D_i$ & $v$ & $\mathcal{B}$  \\ \hline
      \hline \multirow{5}{*}{$D_0$} & $v_{0}$ &  {${\langle}v_{0}, 0{\rangle}$}  \\
      \cline{2-3} & $v_{3}$ &  {${{\langle}v_{0}, 1{\rangle}}$}
      {${\langle}v_{3}, 0{\rangle}$}\\ 
      \cline{2-3} & $v_{6}$ &  {${\langle}v_{0}, 2{\rangle}$}
       {${\langle}v_{2}, 1{\rangle}$}
       {${\langle}v_{6}, 0{\rangle}$}\\
      \cline{2-3} & $v_{8}$ &  {${\langle}v_{0}, 2{\rangle}$}   {${\langle}v_{3}, 1{\rangle}$} 
      {${\langle}v_{6}, 2{\rangle}$}\\
      \cline{2-3} & $v_{9}$ &  {${\langle}v_{0}, 3{\rangle}$} 
      {${\langle}v_{2}, 2{\rangle}$} {${\langle}v_{3}, 2{\rangle}$} {${\langle}v_{6}, 1{\rangle}$} \\
      \hline \multirow{4}{*}{$D_1$} & $v_{1}$ &  {${\langle}v_{0}, 2{\rangle}$}  {${\langle}v_{1},0{\rangle}$}\\
      \cline{2-3} & $v_{4}$ &  {${\langle}v_{0}, 1{\rangle}$}
      {${\langle}v_{1}, 1{\rangle}$}
      {${\langle}v_{4}, 0{\rangle}$}\\
      \cline{2-3} & $v_{5}$ &  {${\langle}v_{0}, 2{\rangle}$}
      {${\langle}v_{1}, 2{\rangle}$}
      {${\langle}v_{3}, 1{\rangle}$}
      {${\langle}v_{4}, 1{\rangle}$}
      {${\langle}v_{5}, 0{\rangle}$}\\
      \cline{2-3} & $v_{10}$ &  {${\langle}v_{0}, 3{\rangle}$} {${\langle}v_{1}, 1{\rangle}$} {${\langle}v_{3}, 2{\rangle}$}
      {${\langle}v_{5}, 1{\rangle}$}\\
      \hline \multirow{3}{*}{$D_2$} & $v_{2}$ &  {${\langle}v_{0}, 1{\rangle}$}
      {${\langle}v_{2}, 0{\rangle}$}\\
      \cline{2-3}  & $v_{7}$ &  {${\langle}v_{0}, 3{\rangle}$}  {${\langle}v_{1}, 1{\rangle}$} 
      {${\langle}v_{2}, 2{\rangle}$} {${\langle}v_{7}, 0{\rangle}$} \\
       \cline{2-3} & $v_{11}$ &  {${\langle}v_{0}, 2{\rangle}$}
       {${\langle}v_{1}, 2{\rangle}$}
       {${\langle}v_{2}, 1{\rangle}$}
       {${\langle}v_{7}, 1{\rangle}$}\\
       \cline{2-3} & $v_{12}$ &  {${\langle}v_{0}, 3{\rangle}$}
       {${\langle}v_{2}, 3{\rangle}$}
       {${\langle}v_{3}, 2{\rangle}$}
       {${\langle}v_{6}, 2{\rangle}$}\\
      \hline 
      \end{tabular}}
      \setcounter{subfigure}{0}
      % may insert caption text here instead
      \subfigure[Road network and district]{\hbox to 0.5\textwidth{\vbox to 0mm{\vfil\null}\hfil}}
      \subfigure[Border Labeling Index]{\hbox to 0.49\textwidth{\vbox to 0mm{\vfil\null}\hfil}}
      \caption{An example of border labeling }
      \label{Fig::labelexample}
\end{figure}

\begin{definition}[District, $D_i$]
We decompose a network $G(V_G, E_G)$ into $m$ mutually exclusive districts $D_i(V_{D_i}, E_{D_i}), i\in[0,m)$, i.e., $\bigcup V_{D_i} = V_G$ and $V_{D_i} \cap V_{D_j} = \varnothing$ for any $i,j\in[0,m)$.
\end{definition}

\begin{definition}[Border Vertex Set of $D_i$, $B_i$]\label{def:border}
A vertex $b \in V_{D_i}$ is a border vertex of district $D_i$ if and only if there exists an edge $(b,v)\in E_G$ where $v \notin V_{D_i}$. The border vertex set of $D_i$ is denoted as $B_i$.
\end{definition}

\begin{example}
Figure~\ref{Fig::labelexample}(a) shows us an example of a network divided into 3 different districts $D_0$, $D_1$ and $D_2$. Vertex $v_0$ is a border of district $D_0$ and vertex $v_1$ is the border of district $D_1$.
\end{example}

\section{Border based Hub Pushing}\label{sec:method:BL}

\subsection{Border Pushing}

Suppose $G(V,E)$ has been divided into $m$ districts, i.e., $D_{0},D_{1},\ldots,D_{m-1}$, and each district has completed its 2-hop cover labeling construction $L_i$. It is crucial to emphasize that the label set $L_i$ can accurately answer distance queries between vertices $s$ and $t$, $\lambda(s,t,L_i)$, only if every single edge of the shortest path lies within the district $D_i$. To resolve this challenge, previous methods have suggested utilizing supplementary data structures such as G-tree\cite{GT15}, Hierarchical Graph Partition (HGP) tree\cite{DHL2023}, and Boundary Tree\cite{boundary22} or expensive online searching method\cite{par16} to help answer queries where shortest paths across different districts. However, these solutions not only raise the construction cost but also increase the query processing overhead.

To overcome the challenges, we propose a method called \textbf{border labeling} in this work. In the majority of hub pushing algorithms~\cite{PLL13,ye_order}, any individual vertex can serve as the {\em hub vertex} to satisfy the 2-hop cover property. However, the main focus of these algorithms is to minimize the size of the labeling set $L$, which may not be the most suitable approach for partitioning-based solutions. In contrast, our proposed solution strategically employs the border vertices as hub vertices to facilitate efficient traversal between vertices across different districts. This approach is based on the understanding that the borders naturally serve as bridges between the districts, and are necessary for queries across districts (since their shortest paths inevitably pass through the borders).

\begin{algorithm}[!htb]
\label{algo:BFP}
\LinesNumbered
\caption{Border Labeling}\label{alg:bl}
%\SetKwFunction{PUSH}{PUSH}%
\KwIn{A list of borders $B$ to be pushed,  road network $G$}
\KwOut{Border labels $\mathcal{B}$}
Initialize an empty set $\mathcal{B}$ for border labels.

\For{$b\in B$}
{
    Create a new queue $Q$ with an initial element ${\langle}b,0{\rangle}$. 
    
    \While{$Q$ is not empty} 
    {
        Dequeue ${\langle}v,d{\rangle}$ from $Q$.
        
        \If{$\lambda(v,b,\mathcal{B}(b)) > d$} 
        {
            Update label  ${\langle}b,d{\rangle}$ onto $\mathcal{B}(v)$.
             
            \For{\textbf{each} unvisited neighbor vertex $u$ of $v$} 
            {
               Enqueue ${\langle}u,d+d_G(u,v){\rangle}$ into $Q$.
            }
        }
        \Else {
            \textbf{continue}\quad\tcp{Pruning}
        }
    % \For{}
}
    
%$\mathcal{B}^{i+1} \gets \mathcal{B}^{i}$
}

\end{algorithm}

Algorithm~\ref{algo:BFP} shows the pseudocode of our border pushing algorithm. Our border pushing algorithm is based on the 2-hop cover approach with hub pushing. Suppose the border set as $B = B_0 \cup \cdots \cup B_{m-1} = \{b_0, b_1, \ldots, b_{q-1}\}$. 
We start with an empty index $\mathcal{B}$, where $\mathcal{B}(u)$ indicates the border label of vertex $u$. 

We iteratively perform Dijkstra's algorithm with pruning from each border vertex, following the processing order determined by the global vertex order $\mathcal{O}$ as defined in the hub pushing algorithm~\cite{PLL13}. 
%in the order of $b_0, b_1, \ldots, b_{m-1}$.
We apply the same pruning strategy when pushing the border vertex $b$. Specifically, we insert the label from $b$ to a reachable vertex only if $\lambda(b, u, \mathcal{B}(u)) > d_G(b, u)$. If this condition is not met, we stop traversing any edge from vertex $u$ and prune it, and we do not add $\langle b, d_G(b,u) \rangle$ to $\mathcal{B}(u)$.

\begin{theorem}\label{thm:border}
%If origin $s$ and destination $t$ can obey at least one of the following two constraints, after 
The border pushing algorithm can guarantee the correctness of the shortest path distance, i.e., $\lambda(s,t,\mathcal{B}) = d_G(s,t)$, under the condition that one of the following constraints is satisfied:
(1) $s \in B \wedge t \in B$ or (2) $s \in D_i \wedge t \in D_j$ where $i\neq j$.
 
\end{theorem}

\begin{proof}
Suppose we obtain $\mathcal{B}$ by executing Algorithm~\ref{alg:bl} with borders pushed in the order of $b_0,...,b_{q-1}$. The border label without applying the pruning strategy, denoted as $\mathcal{B}'$, makes it easier to prove Theorem \ref{thm:border}, as it includes distances of all border vertices in each vertex: 
\begin{equation}
\label{eq:allforone}
    \mathcal{B}'(v) = \{{\langle}b_{0}, d_G(v,b_{0}){\rangle},..., {\langle}b_{q-1}, d_G(v,b_{q-1}){\rangle} \},\quad\forall v\in G.
\end{equation}
Thus we try to use $\mathcal{B}'$ as a bridge to prove. 
Due to   $ \lambda(s,t,\mathcal{B}) = \lambda(s,t,\mathcal{B}')$ by\cite[Theorem~4.1]{PLL13}, we can prove $\lambda(s,t,\mathcal{B}) = d_G(s,t)$ by proving $\lambda(s,t,\mathcal{B}') = d_G(s,t)$.

We then prove $\lambda(s,t,\mathcal{B}') = d_G(s,t)$ by considering the two constraints separately.

For Constraint 1 ($s \in B \wedge t \in B$), as ${\langle}s, 0{\rangle} \in \mathcal{B}'(s)$  and ${\langle}s, d_G(s,t){\rangle} \in \mathcal{B}'(t)$, we will have
\begin{equation}
    \lambda(s,t,\mathcal{B}') = 0 +  d_G(s,t) = d_G(s,t).
\end{equation}

For Constraint 2 ($s \in D_i \wedge t \in D_j$ where $i\neq j$), as an axiom, the shortest path corresponding to $d_G(s,t)$ must pass through at least one vertex in $B$, which means $d_G(s,t) = \min_{b \in B}(d_G(s,b) + d_G(b,t))$.
According to (\ref{eq:allforone}), \begin{equation}
\langle b,d_G(s,b) \rangle \in\mathcal{B}'(s) \wedge \langle b,d_G(t,b) \rangle \in\mathcal{B}'(t),\quad\forall b\in B.
\end{equation}
So 
\begin{equation}
\begin{aligned}
  \lambda(s,t,\mathcal{B}') = \min_{b \in B}(d_G(s,b) + d_G(b,t)) =d_G(s,t).
\end{aligned}
\end{equation}

If any of the constraints is satisfied, $\lambda(s,t,\mathcal{B}') = d_G(s,t)$, which leads to 
 $\lambda(s,t,\mathcal{B}) = d_G(s,t)$.

\end{proof}

\begin{figure}[!hbt]
   
    \centering
    \subfigure[Pushing from $v_0$.]{\includegraphics[width=0.325\textwidth]{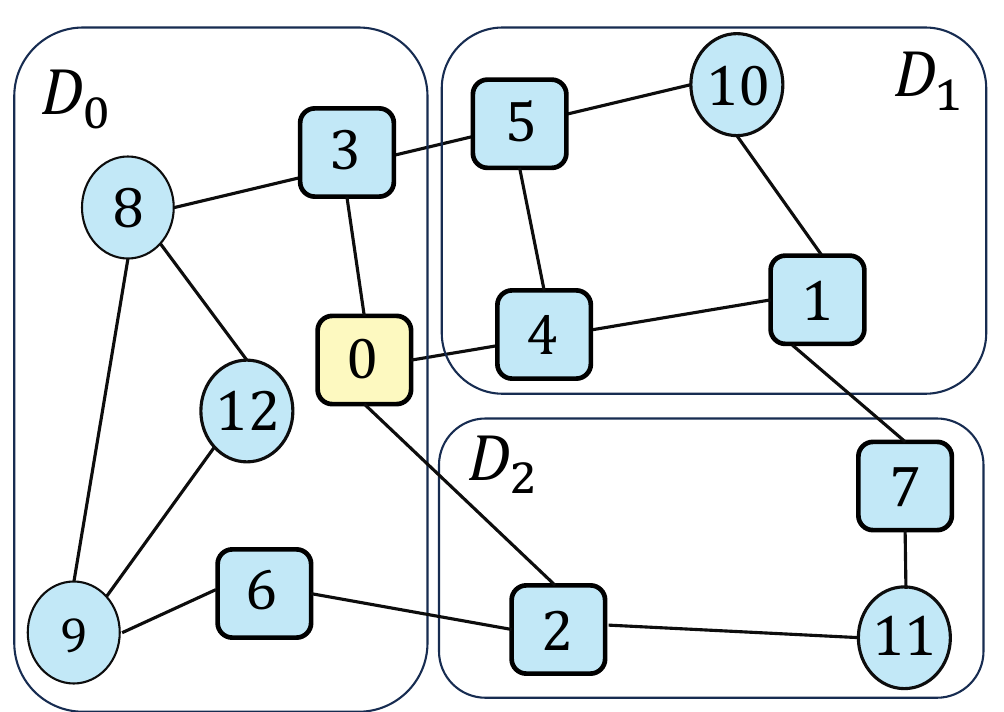}} 
    \subfigure[Pushing from $v_1$.]{\includegraphics[width=0.325\textwidth]{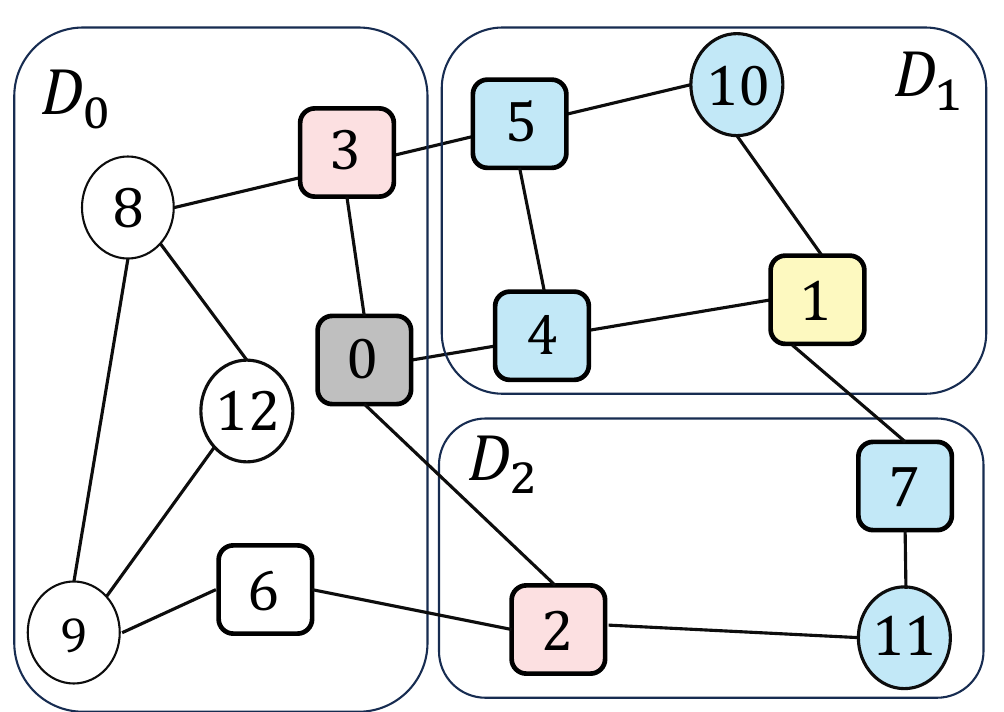}}
    \subfigure[Pushing from $v_2$.]{\includegraphics[width=0.325\textwidth]{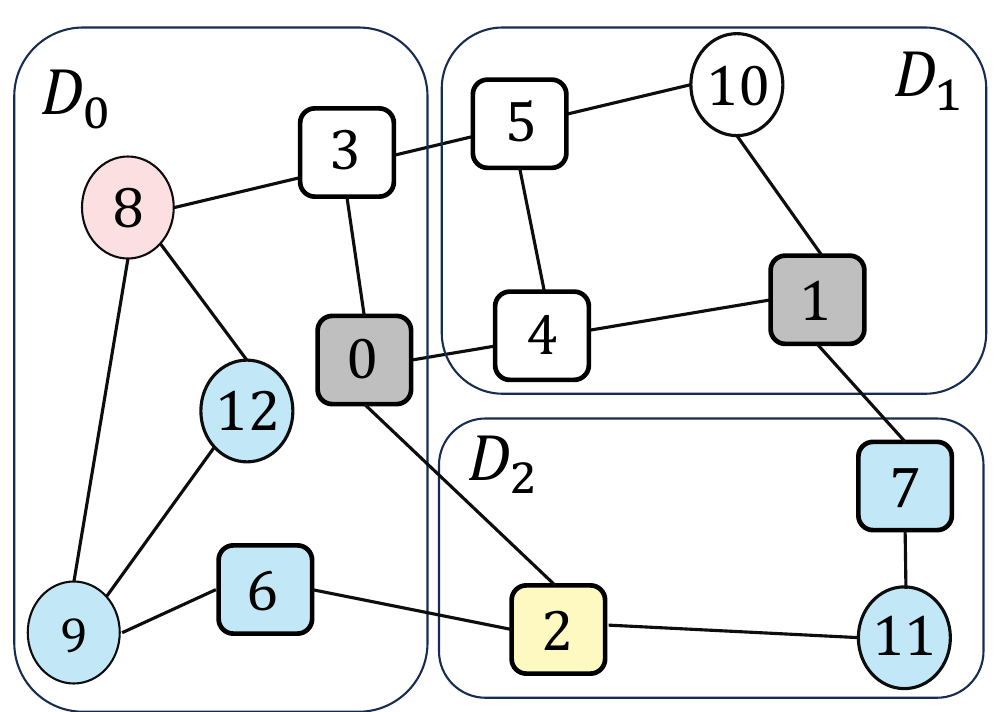}}
    
    \subfigure[Vertices covered by existing labels were pruned.]{\includegraphics[width=0.325\textwidth]{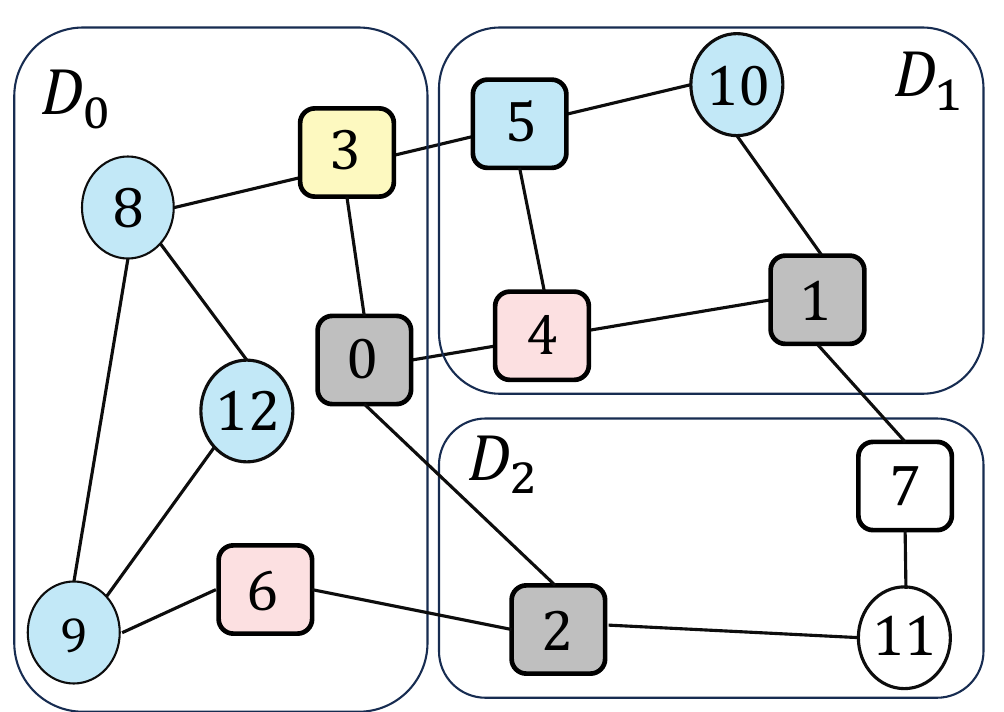}}
      \subfigure[Several borders pushed, the search space limited. ]{\includegraphics[width=0.325\textwidth]{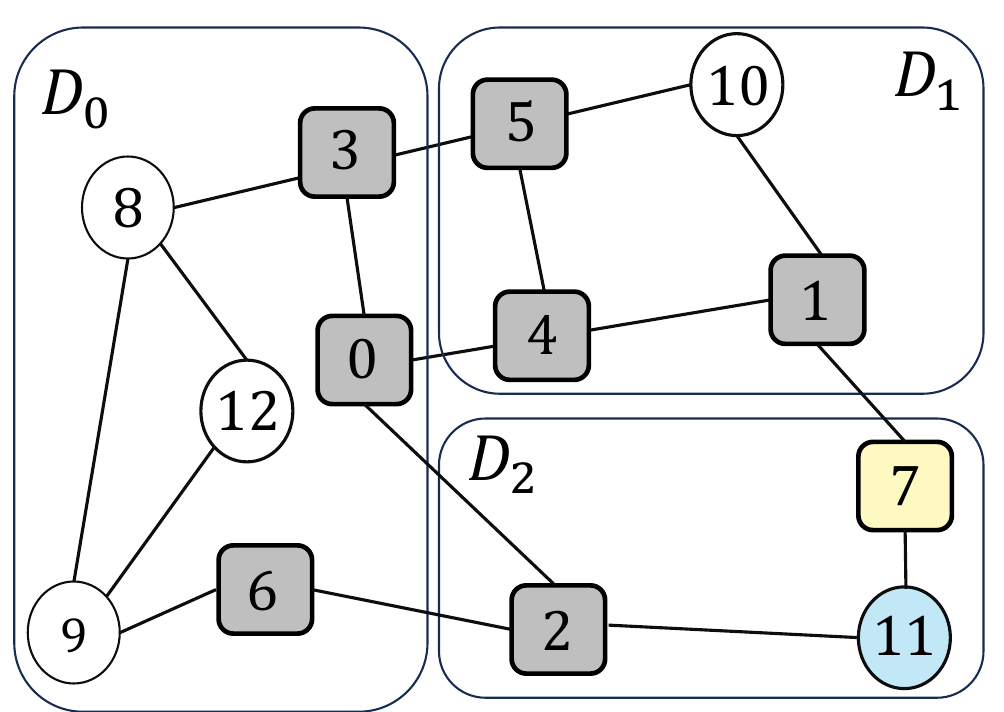}} 
    \subfigure[Our algorithm stopped as we pushed all the borders.]{\includegraphics[width=0.325\textwidth]{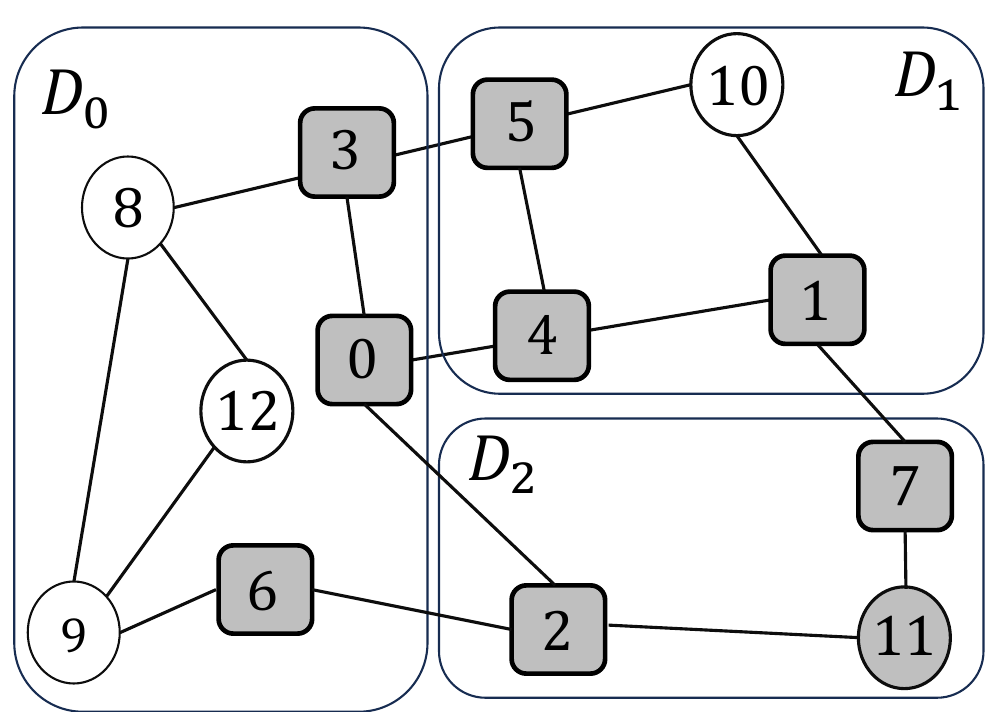}}
  
\caption{Examples of the border labeling algorithm.}
\label{fig::algobl}
\end{figure}

\begin{example}
The border label set $\mathcal{B}$ of the running example (Figure~\ref{fig::algobl}) is shown in Figure~\ref{Fig::labelexample}(b). We can use the border label set $\mathcal{B}$ to answer queries $query(s,t)$ for vertices $s$ and $t$ located in different districts. For example, we can answer $query(v_{11},v_{12})= 4$ using the labels {${\langle}v_{2}, 1{\rangle}$} in $v_{11}$ and {${\langle}v_{2}, 3{\rangle}$} in $v_{12}$, and $query(v_9,v_{10})= 4 $ can be answered by the labels {${\langle}v_{3}, 2{\rangle}$} in $v_9$ and {${\langle}v_{3}, 2{\rangle}$} in $v_{10}$. Additionally, the border labeling algorithm allows us to answer queries involving border vertices in the same or different districts. For example, $query(v_{0},v_{6})= 2 $ can be answered by the labels {${\langle}v_{0}, 0{\rangle}$} in $v_{0}$ and {${\langle}v_{0}, 2{\rangle}$} in $v_{6}$, and $query(v_5,v_{7})=3 $ can be answered by the labels {${\langle}v_{1}, 2{\rangle}$} in $v_5$ and {${\langle}v_{1}, 1{\rangle}$} in $v_{7}$.
\end{example}

\subsection{Border Auxiliary Shortcuts}

Obviously, the border labels $\mathcal{B}$ cannot be utilized to answer a $query(s,t)$ where the non-border origin $s$ and non-border destination $t$ are located in the same district $D_i$. For example, the $query(v_8,v_9)$ cannot be correctly answered. We will then elaborate on how to use district-index techniques to respond to such queries and how the border label set assists in constructing the district-index.

According to Theorem~\ref{thm:border}, our border labeling algorithm is precisely able to compute the correct global shortest distances between any two borders efficiently. 
To calculate the distance between two interior vertices, our approach includes the introduction of auxiliary shortcut edges for each pair of borders $(b_i, b_j, \lambda(b_i, b_j, \mathcal{B}))$ within district $D_i$. This results in the creation of a new set of districts called $D^+=\{D^+_0,...,D^+_i,...\}$. Subsequently, we utilize the standard hub pushing technique (e.g., \cite{PLL13}) to construct the label index $L^+_i$ for district $D^+_i$.

The correctness of our approach is demonstrated as follows.

\begin{theorem}\label{shortcut correctness}
When $s$ and $t$ are in same district $D_i$, $d_{G}(s,t) = \lambda(s,t,L^+_i)$.
\end{theorem}
\begin{proof}

Let us assume that we have two vertices, $s \in D_i$ and $t \in D_i$, and the shortest path between them passes through a set of vertices denoted as $V_{sp}$.

(Case A) Every edge along the shortest path lies exclusively within the district $D_i$, implying that $V_{sp} \subseteq V_{D_i}$. For any edge $(v_i, v_j)$ on the shortest path, where $v_i \in V_{sp}$ and $v_j \in V_{sp}$, we have $d_{G}(v_i, v_j) = d_{D_i}(v_i, v_j)$. Additionally, considering $d_{D^+_i}(v_i, v_j) = d_{G}(v_i, v_j)$, we can conclude that $d_{D_i}(v_i, v_j) = d_{D^+_i}(v_i, v_j)$. Consequently, it follows that $\lambda(s, t, L^+_i) = \lambda(s, t, L_i) = d_{G}(s, t)$.

(Case B) The shortest path has a segment outside $D_i$ or say at least exist a vertex $v_k \in V_{sp}$ and $v_k \notin D_i$. Let $d_{G}(s,t) = d_{G}(s,v_k) + d_{G}(v_k,t) <  \lambda(s,t,L_i)$. As $\lambda(s,t,L_i)$ is local shortest distance of $D_i$, $d_{D_i}(s,t) = \lambda(s,t,L_i)$. So we have
\begin{equation}
    d_{G}(s,t) = d_{G}(s,v_k) + d_{G}(v_k,t) <  d_{D_i}(s,t).
\end{equation}

Under this situation, the shortest path must pass at least two inner borders of $D_i$, suppose they are  $b_m \in V_{sp}$ and $b_n \in V_{sp}$:
\begin{equation}
    d_{G}(s,t) = d_{G}(s,b_m) + d_{G}(b_m,v_k) + d_{G}(v_k,b_n) + d_{G}(b_n,t).
\end{equation}

Given the fact that $D^+_i$ include the shortcuts of every border vertices. Thereby, we have
\begin{align}
    d_{D^+_i}(b_m,b_n) &= d_{G}(b_m,v_k) + d_{G}(v_k,b_n),\\
    d_{D^+_i}(s,b_m) &= d_{G}(s,b_m),\\
    d_{D^+_i}(b_n,t) &= d_{G}(b_n,t).
\end{align}

In conclusion, we obtain $d_{D^+_i}(s, t) = d_{G}(s, t)$. As per the definition, we have $d_{G}(s, t) = \lambda(s, t, L^+_i)$.

\end{proof}

\section{Edge Computing Environments}\label{sec:enivornments}

In this section, we will delve into the integration of our border labeling technique within edge computing environments. Our aim is to enhance the responsiveness of user queries by promoting efficient collaboration among the devices in the system, particularly in {\em real-world data update scenarios}.

\subsection{System Architecture}\label{sec:System}

%In this section, we will present a comprehensive explanation of our approach for integrating the border labeling algorithm with an edge computing system framework, aiming to enhance the responsiveness of user queries through good cooperation of the devices in the system. 

The IBM blog posts ``Architecting at the Edge''~\cite{ibm} and Microsoft Azure's IoT platform \cite{microsoft} provides insightful information on edge computing, covering topics such as devices, servers, benefits, challenges, and applications. Recent research affirms that edge computing brings computation and data storage closer to the network edge, resulting in faster response times, reduced latency, and improved real-time decision-making \cite{TMC4,TMC1}. It also reduces bandwidth usage and enhances privacy and reliability by processing data locally on edge devices or servers \cite{TMC1}.

Our system, based on edge computing architecture, incorporates with multiple layers of computational resources. The topmost layer, known as the computing center layer, utilizes border labels to handle queries across different districts within a city. It oversees and manages the lower layer, which consists of edge servers responsible for specific geographical regions. Each edge server builds a local index and handles queries within its designated region. The size of each region may vary based on factors such as geographical division, networking technology, and computing power. The lowest layer comprises end-user devices like smartphones or in-vehicle navigation systems, connected to the edge servers via 5G. User queries are submitted from the end-user device to the connected edge server, and we elaborate on the query answering process in Section~\ref{sec:method:distance}.

\subsection{Query Processing in Edge Computing Environments}\label{sec:method:distance}

In real-world scenarios, road networks are subject to dynamic updates~\cite{dasfaa23,UE, DBLP:conf/icde/ZhangLHMCZ21, DBLP:conf/sigmod/WeiWL20,dynamic1}. When a significant portion of edges in a city road network undergo frequent updates, it becomes crucial to accurately and efficiently respond to client distance queries. To address this challenge, we leverage edge computing techniques. This involves the collaboration between the cloud computing center and edge servers, where they independently perform computations to construct indexes and handle client distance queries. This collaborative effort ensures the prompt and accurate processing of queries, even in scenarios where the road network undergoes frequent updates.
 
The edge server responsible for the district where the client is located receives the user's distance query request through a 5G signal. Subsequently, the edge server determines whether it can directly handle the distance query or if it needs to forward it to the cloud computing center. This decision is based on the following rules, which govern its actions:

\begin{enumerate}[(1)]
\item \textbf{Origin $s$ and destination $t$ are both located within the district hosted by this edge server.} The edge server will handle the distance query directly.

\item \textbf{Origin $s$ and destination $t$ are both located within a district hosted by another edge server.} The query will be forwarded to that edge server through the computing center. In this scenario, the computing center serves as a forwarding agent, facilitating the communication between the edge servers and ensuring that the query reaches the appropriate server.

\item \textbf{Origin $s$ and destination $t$ are in different districts.} The query will be forwarded to the cloud computing center for handling. In other words, the cloud computing center takes charge of processing and responding to queries that involve vertices from multiple districts.
\end{enumerate}

We explain how the computing center and the edge servers work in detail.

\stitle{Computing Center}
After a period (e.g., one minute), the computing center will ask the edge servers to provide the new traffic situation, including the new edge weights collected from edge IoT devices such as smart surveillance cameras.
The computing center undertakes the reconstruction of the border label $\mathcal{B}$. Once $\mathcal{B}$ is constructed, the computing center becomes capable of responding to client queries that are forwarded to it via $\mathcal{B}$. It then relays the answer back to the respective edge server, acting as an intermediary to relay the response to the user. Simultaneously, the computing center is responsible for forwarding the \textit{Border Auxiliary Shortcuts} to the corresponding edge servers, ensuring efficient and accurate query processing throughout the system.

\stitle{Edge Servers} Each edge server actively collects updated traffic information, which includes new edge weights obtained from edge IoT devices such as smart surveillance cameras. This continuous data collection ensures that the edge servers have the most up-to-date information regarding the traffic situations within their respective districts. 
According to Theorem~\ref{shortcut correctness}, edge server can correctly answer the shortest distance only if receiving all its \textit{Border Auxiliary Shortcuts} from the computing center. In addition, we acknowledge the possibility that the construction of $\mathcal{B}$ from the computing center may take a longer time, despite its relatively shorter computational overhead compared to existing approaches~\cite{tmc5}. This delay could result in some user queries remaining unanswered. To address this concern, we have realized the potential of utilizing a local bound approach. By leveraging the local label index of each district itself $L_i$, we can respond to certain queries without relying on the completion of $\mathcal{B}$. This approach allows for timely responses to queries and mitigates the risk of query delays.

\begin{definition}[Local Bound, $LB(s,t,L_i,B_i$]\label{def:MED}
The local bound of the vertices $s,t \in D_i$ is defined as following:
\begin{equation}\label{eq:MED}
LB(s,t, L_i, B_i)=\min_{b_{i, j} \in B_i} \lambda(s,b_i,L_i) + \lambda(b_j,t,L_i).
\end{equation}
\end{definition}

\begin{theorem}
If $\lambda(s,t,L_i) \leq LB(s,t,L_i,B_i)$,  $d_G = \lambda(s,t,L_i)$.
\end{theorem}

\begin{figure}[!hbt]
    \begin{minipage}[b]{.99\linewidth}
    \centering
    \subfigure[]{\includegraphics[width=0.243\textwidth]{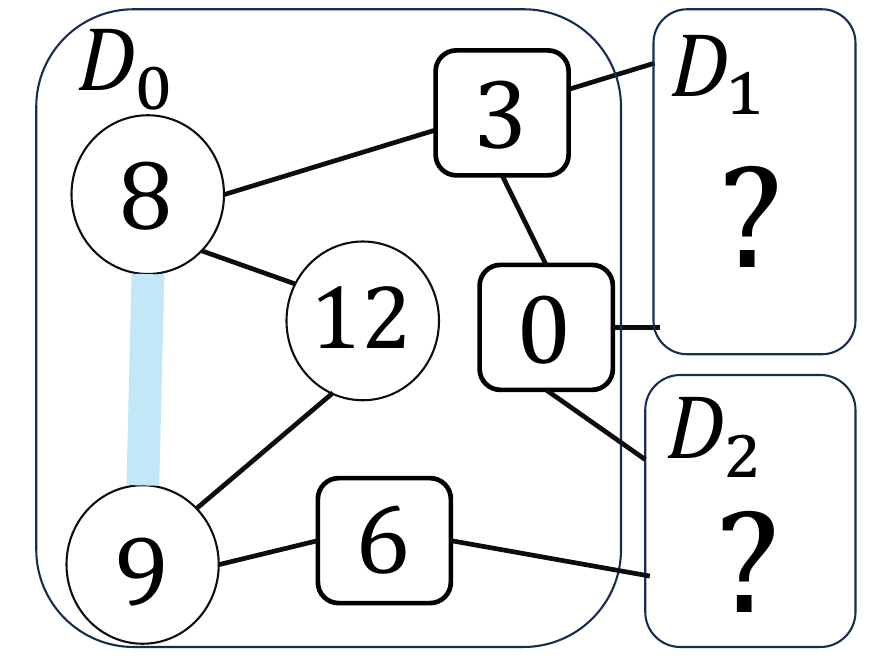}} 
    \subfigure[]{\includegraphics[width=0.243\textwidth]{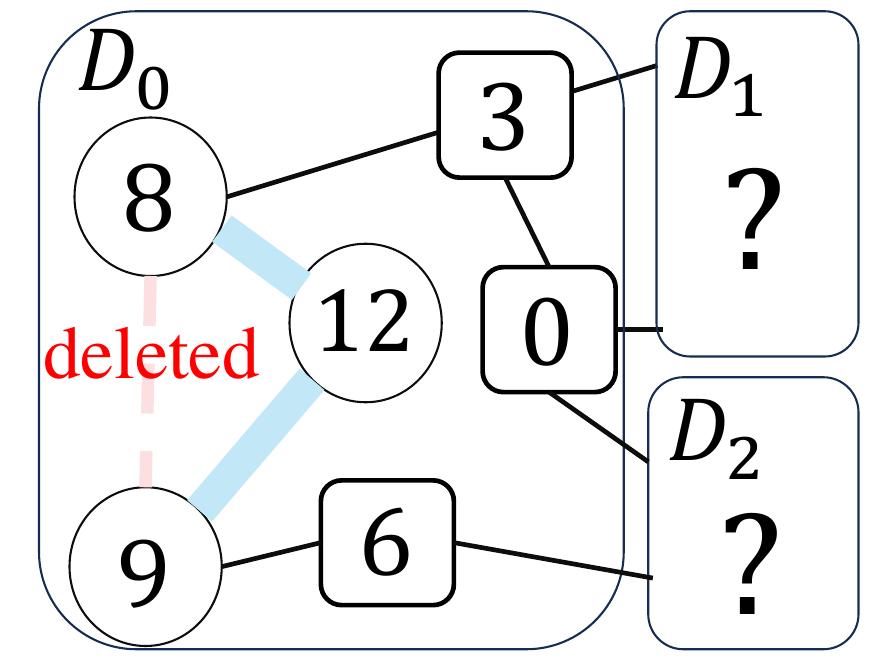}}
     \subfigure[]{\includegraphics[width=0.243\textwidth]{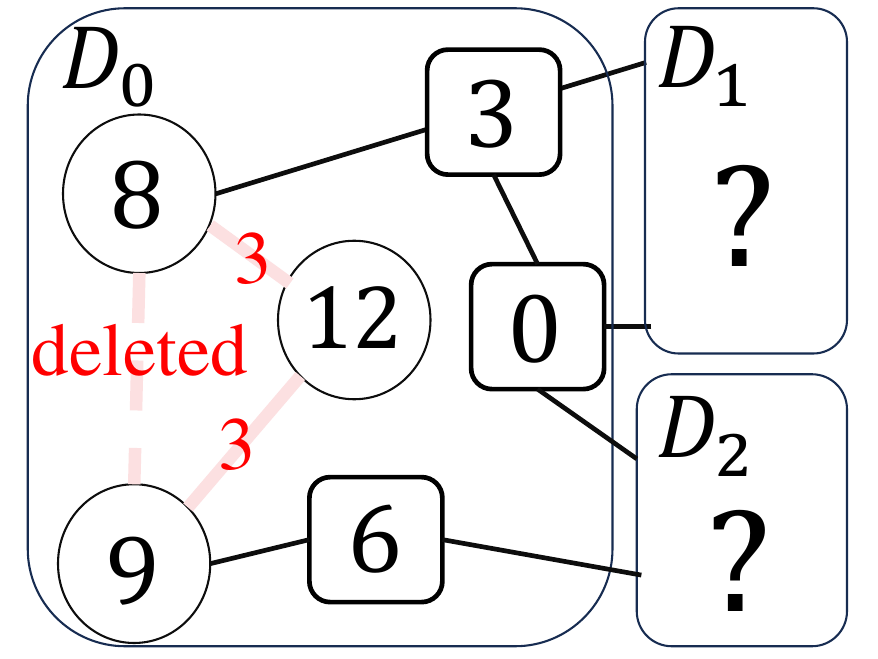}}
     \subfigure[]{\includegraphics[width=0.243\textwidth]{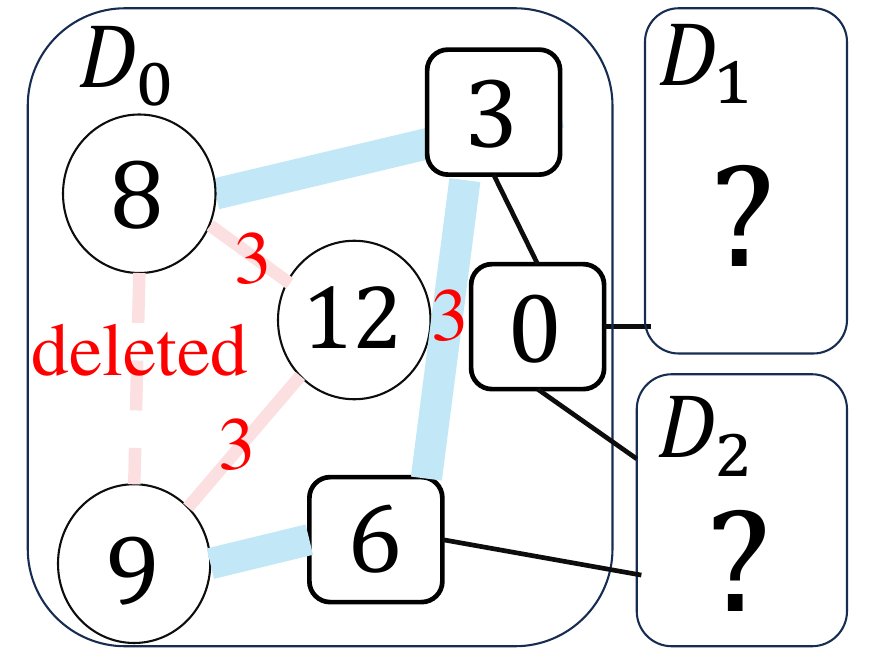}}

\end{minipage}\hfill
  \begin{minipage}[b]{.48\linewidth}
     \scriptsize
     \centering
     \begin{tabular}{|c|l|}
      \hline  $(c)$ & {\bf Intra Index $L_0$}\\ \hline
      \hline $v_{0}$ &  {${\langle}v_{0}, 0{\rangle}$}  \\
      $v_{3}$ &   {${\langle}v_{0}, 1{\rangle}$}
      {${\langle}v_{3}, 0{\rangle}$}\\
        $v_{6}$ &  {${\langle}v_{0}, 9{\rangle}$} {${\langle}v_{3}, 8{\rangle}$}
        {${\langle}v_{6}, 0{\rangle}$}\\ 
       $v_{8}$ &  {${\langle}v_{0}, 2{\rangle}$}
       {${\langle}v_{3}, 1{\rangle}$}
       {${\langle}v_{6}, 7{\rangle}$} 
       {${\langle}v_{8}, 0{\rangle}$} \\
       $v_{9}$ &  {${\langle}v_{0}, 8{\rangle}$} 
      {${\langle}v_{3}, 7{\rangle}$} {${\langle}v_{6}, 1{\rangle}$} 
      {${\langle}v_{8}, 6{\rangle}$}
      {${\langle}v_{9}, 0{\rangle}$}
      \\
      $v_{12}$  &
      ...\\

      \hline 
      \end{tabular}
      
  \end{minipage}\hfill
  \begin{minipage}[b]{.43\linewidth}
      \scriptsize
      \begin{tabular}{|c|l|}
      \hline  $(d)$ & {\bf Intra shortcut Index $L^+_0$}\\ \hline
      \hline $v_{0}$ &  {${\langle}v_{0}, 0{\rangle}$}  \\
      $v_{3}$ &   {${\langle}v_{0}, 1{\rangle}$}
      {${\langle}v_{3}, 0{\rangle}$}\\
       $v_{6}$ &  {${\langle}v_{0}, 4{\rangle}$} {${\langle}v_{3}, 3{\rangle}$} {${\langle}v_{6}, 0{\rangle}$} \\
       $v_{8}$ &  {${\langle}v_{0}, 2{\rangle}$} {${\langle}v_{3},1{\rangle}$} {${\langle}v_{8}, 0{\rangle}$} \\
       $v_{9}$ &  {${\langle}v_{0}, 5{\rangle}$} 
      {${\langle}v_{3}, 4{\rangle}$} {${\langle}v_{6}, 1{\rangle}$} {${\langle}v_{9}, 0{\rangle}$} \\
      $v_{12}$ &  ...\\

      \hline 
      \end{tabular}
  \end{minipage}
  \caption{Example of Index in intra-query processing under dynamic scenario. The blue curve denotes the shortest path.}
  \label{fig:dynamic}
\end{figure}
\section{Performance Studies}\label{sec:exp}
In real-world GIS service applications the capacity to rapidly respond to client queries based on the latest traffic information is a critical metric. However, before conducting tests on real dynamic scenarios, it is necessary to evaluate the fundamental static performance of various algorithms. For instance, some algorithms may offer remarkable speed but huge storage, making it impractical in industry scenarios. For this reason, we evaluate the construction and response time, as well as the index size on 10 different scale road networks from a public dataset\footnote{\url{https://www.diag.uniroma1.it/challenge9/index.shtml}}. 

All algorithms were implemented with C++ and compiled by g++ with -O3 flag. All experiments were conducted on a Linux Sever in 64-bit Ubuntu 22.04.3 LTS with 2 Intel Xeon E5-2696v4 and 128GB main memory. We omitted the result of a method if it ran out of memory of our machine or did not terminate within 1 hours and denote them as \textbf{MLE} (memory-limit exceeded) and \textbf{TLE} (time-limit exceeded).  We use a 32-bit integer to represent a vertex ID or a distance value in the index. We report the index size for each evaluated method. For hub pushing based techniques, each label is a 2-tuple $\langle hub,dist\rangle$ for those solutions only applicable to distance queries. Our codes can be found in \url{https://anonymous.4open.science/r/Submit-anonymously/}.

\begin{table}[htbp]
\centering
\caption{Road Networks}
\label{tbl:dataset}
\resizebox{1\textwidth}{!}{ % Adjust the resize values as needed
\begin{tabular}{|c|c|c|c|c|c|c|c|}
\hline
Graph Name & $|V|$ & $|E|$ & Size & Graph Name & $|V|$ & $|E|$ & Size \\ \hline
	New York City (NY) & 264K & 773K & $3.6$MB & San Francisco Bay Area (BAY) & 321K & 800K & $4.0$MB \\\hline
Colorado (COL) & 436K & 1M & 5.6MB & Florida (FLA) & 1M & 2.7M & 14MB \\\hline
Northwest USA (NW) & 1.2M & 2.8M & 16MB & Northeast USA (NE) & 1.5M & 3.9M & 21MB \\\hline
California and Nevada (CAL)	
 & 1.9M & 4.7M & 26MB & Great Lakes (LKS) & 2.8M & 6.9M & 39MB \\\hline
 	Eastern USA (E)	
 & 3.5M & 8.8M & 50MB & Western USA (W) & 6.2M & 15.2M & 88MB \\\hline
\end{tabular}
}
\label{table:graphs}
\end{table}

\subsection{Algorithm Performance Study}

\begin{figure}[!hbt]
  \centering
  \includegraphics[width=1.0\linewidth]{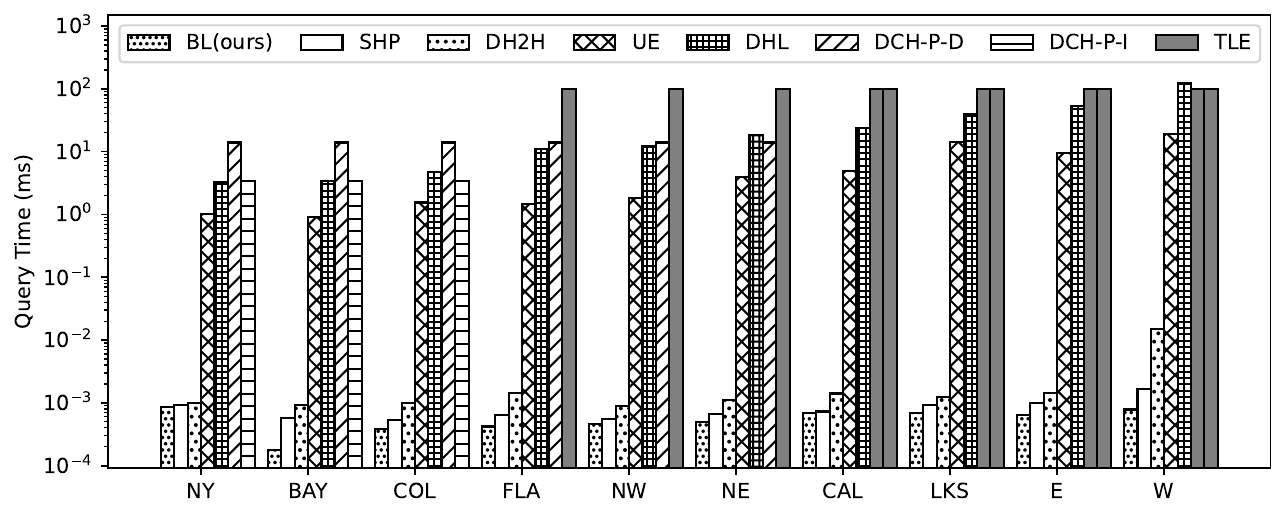}
  \caption{Query Processing}
  \label{fig:query}
\end{figure}

\stitle{Response Time} We test 100,000 random queries for each dataset in terms of response time shown in Fig.\ref{fig:query}. Methods based on Hub Labeling, such as our approach, SHP, and DH2H, which operate at the microsecond level, are significantly ahead of CH-based methods that perform at the millisecond level. The reason our method can outpace traditional Hub Labeling-based methods is that we expedite queries by {\em organizing} them into mutually independent and smaller search spaces. For example, the average label size of a border label does not exceed the number of borders, thereby reducing the cost of linear merging.

\begin{table*}[!htb]
\centering
\caption{Indexing time and index size. }
	\centering
	\resizebox{1.0\textwidth}{!}{
		\begin{tabular}{|c|c|c|c|c|c|c|c|c|c|c|c|c|c|c|c|c|c|} \hline
			{\multirow{4}{*}{Graph}} & \multicolumn{8}{c|}{Indexing Time ($s$)}  &\multirow{13}{*}{~}  & \multicolumn{8}{c|}{Index Size ($MB$)}   \\  \cline{2-9}\cline{11-18}
			
			&  \multicolumn{2}{c|}{Ours}  & \multicolumn{6}{c|}{Competitors} & & \multicolumn{2}{c|}{Ours}  & \multicolumn{6}{c|}{Competitors} \\ 
            \cline{2-9}\cline{11-18}

			& \multirow{2}{*}{BL} & \multirow{2}{*}{Districts} & \multirow{2}{*}{SHP} & \multirow{2}{*}{UE} & \multirow{2}{*}{DH2H} & \multirow{2}{*}{DCH-P-D} & \multirow{2}{*}{DCH-P-I} &\multirow{2}{*}{DHL} & &\multirow{2}{*}{BL-INT} & \multirow{2}{*}{BL-INN} & \multirow{2}{*}{SHP} & \multirow{2}{*}{UE} & \multirow{2}{*}{DH2H} & \multirow{2}{*}{DCH-P-D} & \multirow{2}{*}{DCH-P-I} & \multirow{2}{*}{DHL}\\
			&& &  &  &  &  &  &  &  & & &  & &  &  & & \\  \cline{1-9}\cline{11-18}
			  NY &  0.7 &  10.8 & 13.4 & 3.9 & 9.6  & 145.2 & 1,676.1 & 6.3 & &  13 &  246 & 303 & 38 & 391 & 7 & 7 & 15\\  \cline{1-9}\cline{11-18}
           BAY & 1.3 &  3.8 & 6.4 & 3.2 & 5.8 & 155.7 & 1,192.8 & 3.4 & &  28 &  116 & 177& 31 & 377 & 7 & 7 &12\\  \cline{1-9}\cline{11-18}
			 COL & 3.2 & 6.2 & 8.9 & 3.8 & 6.8 & 239.1 & 2,200.1 & 4.2 & & 49 & 140 & 173 & 39  & 587 & 8 & 8 & 15 \\ \cline{1-9}\cline{11-18}
			 FLA & 9.2 &  18.4 & 32.9 & 9.1 & 16.2 & 1,3977.8  & TLE & 9.8 &  &  165 & 432 & 710 & 100 & 1330  & 21 & TLE & 39\\  \cline{1-9}\cline{11-18}
			 NW & 14.9 & 19.7  & 25.8 & 10.5  &  18.6 & 1,559.7 & TLE & 10.6 &   &  195 & 479 & 572 & 99 & 1675 & 21 & TLE  &38

    \\  \cline{1-9}\cline{11-18}
    NE & 15.1 &  32.2 & 64.3 & 24.7
   & 42.3 & 3,491.3 & TLE &26.6 & & 240 &  755 & 919  & 171 & 3,152 & 33 & TLE  &67\\  \cline{1-9}\cline{11-18}
    CAL & 22.0 &  48.5 & 53.3 & 23.3 & 42.9  & 5,277.4 & TLE &25.6 & & 348 &  844 &1349 & 184 & 3,999 & 902 & TLE &71\\  \cline{1-9}\cline{11-18}
			 LKS & 14.2 &  85.5 & 120.4 & 76.7 &  106.5 & TLE & TLE & 58.2 & & 215 & 1358  & 1,600 & 318 & 8,885 & TLE & TLE  & 122\\  \cline{1-9}\cline{11-18}
			 E & 68.2 &  95.3 & 166.0 & 54.7 & 116.1 & TLE & TLE &53.0 & &733 &  1953 & 3360 & 184 & 9,917 &  TLE  & TLE &136\\ \cline{1-9}\cline{11-18}
			 W & 125.8 &  148.4 & 238.5 & 94.9 & 181.7 & TLE & TLE &89.7 & & 2070	&  3251
    & 5550 & 590 & 21,076 & TLE & TLE & 227 \\ \hline  					
		\end{tabular}
	}\label{tbl:allindexing1}
\end{table*}

\stitle{Indexing time and Index size} As Table \ref{tbl:allindexing1} shows, we firstly study the relative performance of indexing time. Overall, UE and DHL are the most efficient methods with least indexing time. Both approaches have their own advantages and disadvantages when applied to different datasets, although the differences between them are minimal. In the table, our method is represented in two columns: BL and Districts. BL denotes the time taken to establish the border labeling, whereas Districts represent the cumulative time taken to compute the shortcuts using border labeling for each district, in addition to the time taken to sequentially build local indexes for each district. Despite our method falling slightly behind CH-based methods, even when dealing with large datasets like W with 15.2 million edges, we only require approximately 3-4 times the duration of the best-performing approach. This indicates that our method is comparable to the more straightforward and easier-to-maintain approaches, bringing it close in terms of efficiency. At the same time, while another HL-based method, DH2H, is faster than ours, it incurs a significant label size; for instance, for W, it requires 20GB of memory, which can be prohibitive for many systems. Although the two DCH methods lead in terms of index size, the benefits gained are negligible when compared to the enormous construction time costs; they struggle to build indexes for larger road networks within a tolerable timeframe.

In conclusion, our experimental results indicate that our method is characterized by ultra-fast query speed, suitable label size, and competitive construction time. Our approach is highly suitable for application in a wide range of scenarios.

\section{Conclusion}\label{sec:Conclusion}
\label{conclusion}

In this work, we proposed a novel and efficient method to address exact shortest distance queries under dynamic road networks. Our method is based on distance labeling on vertices, which is common to the existing exact distance querying methods, but our labeling algorithm stands on a totally new idea. 
Based on the concept of districts and borders, our algorithm conducts Dijistra from all the border vertices with pruning. Moreover, we also proposed adaptive bound to ensure correctness for in-district local queries. Though the algorithm is simple, our query performance surpasses our competitors significantly decreasing the query latency. Our border pushing order is degree-based, which can save preprocessing time. Furthermore, our method is inherently aligned with the edge computing environment, which demonstrates its potential application value in the industrial scenarios and its capability to accommodate multiple users. In summary, this can represent a highly robust edge computing based database system for a diverse range of applications.
In the future, we would explore the potential of various labeling methods and investigate the possibilities of employing a hybrid ordering approach to pushing, rather than a singular method.

\bibliographystyle{splncs04}
\bibliography{ECSDQ}
\end{document}